\documentclass{article}
\usepackage[hmargin=3cm,vmargin=3cm]{geometry}
\usepackage{pdfpages}
\usepackage{epsfig}
\usepackage{amssymb}
\usepackage{amsthm}
\usepackage{mathtools}
\usepackage[scaled]{helvet}
\usepackage{caption}
\usepackage{algorithm}
\usepackage{enumitem}
\usepackage{algpseudocode}
\usepackage{color}

\newtheorem{lemma}{Lemma}
\newtheorem{theorem}{Theorem}

\newtheorem{remark}{Remark}

\newtheorem{definition}{Definition}

\newcommand{\ecc}{\textsf{\textsc{UCC }}}
\newcommand{\vcc}{\textsf{\textsc{BCC }}}

\newcommand{\ec}{\textsf{\textsc{UCAST-CONGEST }}}
\newcommand{\vc}{\textsf{\textsc{BCAST-CONGEST }}}
\newcommand{\cc}{\textsf{\textsc{CONGEST-CLIQUE }}}

\newcommand{\congest}{\textsf{\textsc{CONGEST }}}
\newcommand{\local}{\textsf{\textsc{LOCAL }}}

\newenvironment{enumerate*}%
  {\vspace{-2ex} \begin{enumerate} %
     \setlength{\itemsep}{-1ex} \setlength{\parsep}{0pt}}%
  {\end{enumerate}}
 
\newenvironment{itemize*}%
  {\vspace{-2ex} \begin{itemize} %
     \setlength{\itemsep}{-1ex} \setlength{\parsep}{0pt}}%
  {\end{itemize}}
 
\newenvironment{description*}%
  {\vspace{-2ex} \begin{description} %
     \setlength{\itemsep}{-1ex} \setlength{\parsep}{0pt}}%
  {\end{description}}

\newcommand{\BO}{\mathcal{O}}

\newcommand{\disj}[1]{\mathop{\mathrm{DISJ_{#1}}}}

\newtheorem*{lem:d2}{\textbf{Lemma \ref{diam2}}}
\newtheorem*{lem:2a}{\textbf{Lemma \ref{2approx2}}}
\newtheorem*{lem:a}{\textbf{Lemma \ref{approxguarantee}}}
\newtheorem*{lem:gvc}{\textbf{Lemma \ref{greedyvc}}}

\begin{document}

\author{
Stephan Holzer
\\ MIT 
\\ Cambridge, MA, USA
\\ \texttt{holzer@mit.edu}
\and 
Nathan Pinsker
\\ MIT 
\\ Cambridge, MA, USA
\\ \texttt{npinsker@mit.edu}
}
\date{}

\title{Approximation of Distances and Shortest Paths in the Broadcast Congest Clique\footnote{Work supported by the following grants: AFOSR Contract Number FA9550-13-1-0042, NSF Award 0939370-CCF, NSF Award CCF-1217506, NSF Award number CCF-AF-0937274.}}

\maketitle
\thispagestyle{empty}

\begin{abstract}
We study the broadcast version of the \cc model of distributed computing~\cite{kuhn2014,lotker2003mst,nanongkai2014distributed}. In this model, in each round, any node in a network of size $n$ can send the same message (i.e. broadcast a message) of limited size to every other node in the network. Nanongkai presented in [STOC'14~\cite{nanongkai2014distributed}] a randomized $(2+o(1))$-approximation algorithm to compute all pairs shortest paths (APSP) in time\footnote{We use the convention that $\tilde{\Omega}(f(n))$ is essentially $\Omega(f(n)/$polylog$f(n))$ and $\tilde{\BO}(f(n))$ is essentially $\BO(f(n) $polylog$f(n))$.} $\tilde\BO(\sqrt{n})$ on weighted graphs. We complement this result by proving that any randomized $(2-o(1))$-approximation of APSP and $(2-o(1))$-approximation of the diameter of a graph takes $\tilde\Omega(n)$ time in the worst case. This demonstrates that getting a negligible improvement in the approximation factor requires significantly more time.
Furthermore this bound implies that already computing a $(2-o(1))$-approximation of all pairs shortest paths is among the hardest graph-problems in the broadcast-version of the \cc model. This is true as any graph-problem can be solved trivially in linear time in this model and contrasts a recent $(1+o(1))$-approximation for APSP that run in time $\BO(n^{0.15715})$ and an exact algorithm for APSP that runs in time $\tilde\BO(n^{1/3})$) in the unicast version of the \cc model \cite{computing2014censor, kaski2014algebraisation}. 

This lower bound in the \cc model is derived by first establishing a new lower bound for $(2-o(1))$-approximating the diameter in weighted graphs in the \congest model, which is of independent interest. This lower bound is then transferred to the \cc model.

On the positive side we provide a deterministic version of Nanongkai's $(2+o(1))$-approximation algorithm for APSP~\cite{nanongkai2014distributed}. To do so we present a fast deterministic construction of small hitting sets. We also show how to replace another randomized part within Nanongkai's algorithm with a deterministic source-detection algorithm designed for the \congest model in \cite{lenzen2013efficient}. 
\end{abstract}

\section{Introduction}
In a distributed message passing model a network is classically represented as a graph. In this graph any node can send (pass) one message to its neighbors in every round. There are two major research directions concerning message passing models. 

The first research direction deals with determining the locality and congestion of problems. Using the \local model~\cite{peleg}, where message-size is unbounded, one tries to characterize the locality of problems, which is the ability of a node to make decisions regarding a problem purely based on information on its local neighborhood in a graph. Using the \congest model~\cite{peleg}, where message-size is bounded, one tries to characterize the delays caused by congestion. Congestion arises due to bottlenecks in the network that do no provide enough bandwidth during the computation. Both, the \local model and the \congest model are classic models that have received a great deal of attention in the past decades. Recently it was pointed out in~\cite{boaz2011sorting} that the \congest model does not avoid interference from locality issues, while the \local model avoids interference from congestion. To be more precise, congestion is completely avoided in the \local model due to unlimited bandwidth. On the other hand the complexity of algorithms in the \congest model may still depend on the local structure of a graph (e.g. lower bounds transfer from the \local model). To truly separate the study of congestion from locality, one needs to consider networks that avoid locality issues. These are e.g. networks in which each node is directly connected to any other node in the network (represented by a clique), which is a network in which any graph problem can be solved within one round in case unlimited bandwidth is available. Such a model was introduced earlier by Lotker et al.~\cite{lotker2003mst} with the intention to study overlay networks that have this property and was coined the \cc model. Examples of parallel systems design that recently provided additional motivation to this original motivation to study the \cc as an overlay network~\cite{lotker2003mst} are included in Section~\ref{sec:rel}.

The second research direction focuses on determining the power of broadcast compared to (multi-)unicast. Broadcast denotes the setting in which a node can only send the same message to all its neighbors at the same time, while in a (multi-)unicast setting each node can send different messages to different neighbors at the same time. 

Results of this paper push both research directions. To be more precise, we present a linear lower bound and new improved bounds for a broadcast model (the \vcc model, see definition below) that purely studies congestion.

\begin{definition}
When applied to the \cc model, we denote by 
\ecc model 
the (multiple-)unicast version of the \cc model, and by 
\vcc model 
the broadcast version of the \cc model~\cite{kuhn2014,lotker2003mst}.
\end{definition}

\subsection{Contribution}
In this context, this paper extends the work of~\cite{kuhn2014,holzer2013phd,nanongkai2014distributed}. Drucker, Kuhn and Oshman~\cite{kuhn2014} started studying the difference in computational power between the \ecc and the \vcc models. Like~\cite{kuhn2014} we present a linear lower bound in the \vcc model. The lower bounds of~\cite{kuhn2014} were the first deterministic (and conditional randomized) linear lower bounds in this model and consider subgraph detection. Ours are the first unconditional randomized linear lower bounds, while we consider $(2-o(1))$-approximations of APSP and diameter. This result demonstrates:
\begin{itemize}
\item There is a huge (at least quadratic) difference in the complexity between computing a $(2+o(1))$-approximation \cite{nanongkai2014distributed} and a $(2-o(1))$-approximation of APSP in this model. 
\item Computing a $(2-o(1))$-approximation of APSP is among the hardest graph-problems in the \vcc model: any graph-problem (with $\BO(\log n)$-encodable weights) can be solved in linear time, as each node is incident to at most $n-1$ edges. An algorithm could let each node $v$ broadcast the IDs of all of $v$'s neighbors and weights of incident edges in time $\BO(n)$. Then each node in the network has full information on the graph and can perform any computation (including e.g. NP-complete problems) internally, which does not contribute to the runtime.
\item One might not be interested in improving the approximation factor below $2-o(1)$, as this takes almost as much time as computing the exact solution.
\item There is a clear separation between the \ecc and \vcc model with respect to APSP computation. Our lower bounds contrast the results of \cite{computing2014censor, kaski2014algebraisation}, who showed that e.g. even exact APSP can be solved in the \ecc model within $\tilde\BO(n^{1/3})$ time and $(1+o(1))$-approximated in time $\BO(n^{0.15715})$. 
\item Any $\alpha$-approximation of the diameter cannot be computed faster than an $\alpha$-approximation to APSP (for any $\alpha, 1\leq \alpha \leq 2-o(1)$) in the \vcc model, while in the \ecc model there currently exists a faster algorithm for exact diameter computation than for exact APSP \cite{computing2014censor, kaski2014algebraisation}. 
\end{itemize}
Note that this lower bound strengthens the $\tilde\Omega(\sqrt{n})$ lower bound for exact computation of APSP in the \vcc model by \cite{computing2014censor} in terms of runtime and extends it to approximations. The authors of \cite{computing2014censor} provided this lower bound independently and simultaneously via matrix multiplication lower bounds.

To obtain our lower bounds, we first use techniques of~\cite{frischknecht2012networks} to derive an $\tilde\Omega(n)$-round lower bound to $(2-o(1))$-approximate the diameter of weighted graphs in the \congest model. This implies an $\tilde\Omega(n)$-round lower bound to $(2-o(1))$-approximate APSP. To prove our lower bounds, we modify a construction for unweighted graphs that was claimed in~\cite{holzer2012optimal} and can also be found in~\cite{wattenhofer2012lecture} to the weighted setting. As in~\cite{frischknecht2012networks}, this construction is used to transfer lower bounds for set disjointness from two-party communication complexity~\cite{kushilevitz97}. Next we transfer this lower bound from the \congest model to the \vcc model. Compared to this,~\cite{kuhn2014} uses lower bounds for the number-on-forehead model (NOF) of communication complexity~\cite{kushilevitz97} in combination with their constructions.

Apart from these lower bounds, we derive positive results on computing APSP by extending the line of work of~\cite{holzer2013phd,nanongkai2014distributed}. We start by replacing the randomized parts of a recent result by Nanongkai \cite{nanongkai2014distributed}, who presents an algorithm in~\cite{nanongkai2014distributed} for a $(2 + o(1))$-approximation of the all-pairs shortest paths problem in the \vcc model in $\tilde\BO(\sqrt{n})$ rounds, with deterministic ones. We show that the resulting algorithm runs in the \vcc model.

\begin{table}
  \begin{minipage}{\textwidth}
  \centering
  \begin{tabular}{ c | c | c | c }
     approx. factor & APSP & Diameter & SSSP \\ \hline
     $1$ & $\BO(n)^{\#}$ & $\BO(n)^{\#}$ & $\tilde\BO(\sqrt{n})^*$ \\ \hline
     $2 - o(1)$ & $\tilde\Omega(n)^{\dag}$  & $\tilde\Omega(n)^{\dag}$ & --- \\ \hline
     $2$ & --- & $\tilde\BO(\sqrt{n})^*$  & --- \\ \hline
     $2 + o(1)$ & $\tilde\BO(\sqrt{n})^{\ddag}$  &  $\tilde\BO(\sqrt{n})^{\ddag}$ & --- \\
  \end{tabular}
\begin{itemize}[noitemsep]
\item[$\#$)] Trivial bound: collect the whole topology in a single node, perform computation internally.
\item[$*$)] Nanongkai's SSSP algorithm~\cite{nanongkai2014distributed}. See Remark~\ref{rem:2diam} for the diameter approximation.
\item[$\dag$)] Our randomized lower bound, see Theorem~\ref{thm:2-approxLB-vcc}.
\item[$\ddag$)] Our deterministic version of the randomized algorithm of \cite{nanongkai2014distributed}, see Theorem~\ref{thm:APSP-vcc}.
\end{itemize}
  \end{minipage}
\captionsetup{singlelinecheck=off}
\caption{Summary of new and previous results for problems we study on positively weighted graphs in the \vcc model. Recent results of \cite{computing2014censor,kaski2014algebraisation} in the \ecc model are summarized in Section \ref{sec:rel}.}
\end{table}

\subsection{Structure of the Paper}
We review related work in Section \ref{sec:rel} and define the computation models and terminology that we work with in Section \ref{sec:model}. Our lower bounds are presented in Section \ref{sec:lb}, where we start with a review of two-party communication complexity, state the lower bounds for the \congest model and transfer them to the \vcc model. A key-ingredient for our upper bounds is a deterministic hitting set construction, which we present in Section \ref{sec:vccHS}. Finally, in Section \ref{sec:vccAPSP}, we present our deterministic version of Nanongkai's all-pairs shortest paths approximation algorithm in the \vcc model. We conclude by briefly mentioning some open problems and directions for future work in Section \ref{sec:open}.

\section{Related Work}\label{sec:rel}
\textbf{Algorithms in the \vcc and \ecc models:} The first to study the \cc model were Lotker et al.~\cite{lotker2003mst}, where they presented an $\BO(\log \log n)$-round algorithm for constructing a minimum spanning tree in the \ecc model. This was improved by Pemmaraju and Sardeshmukh to $\BO(\log \log n)$ in \cite{pemmaraju2014minimum}. Lenzen obtained in~\cite{lenzen2012optimal} an $\BO(1)$-round algorithm in the \ecc model for simultaneously routing $n$ messages per vertex to their assigned destination nodes, as well as an $\BO(1)$ algorithm for sorting $\BO(n^2)$ numbers, given that each vertex begins the algorithm knowing $\BO(n)$ numbers. Independently Patt-Shamir and Teplitsky~\cite{boaz2011sorting} showed a similar, but slightly weaker result on sorting in the \ecc model. Later Hegeman et al.~\cite{DBLP:journals/corr/HegemanPS14} provided constant and near-constant (expected) time algorithms for problems such as computing a 3-ruling set, a constant-approximation to metric facility location, and (under some assumptions) a constant-factor approximations to the minimum spanning tree in the \ecc model. Holzer~\cite{holzer2012optimal} provided an deterministic $\BO(\sqrt{n})$-algorithm for exact unweighted SSSP (equivalent to computing a breadth first search tree) in the \vcc model. Independently Nanongkai~\cite{nanongkai2014distributed} provided randomized (w.h.p.) algorithms in the \vcc model that take $\tilde\BO(n^{1/2})$ rounds to compute (exact) SSSP, and $\tilde\BO(n^{1/2})$ rounds to $(2 + o(1))$-approximate APSP on weighted graphs. Much of our work for deterministic APSP builds off~\cite{nanongkai2014distributed}, primarily on his idea of "shortcut edges", which do not change the weighted shortest path length between any two nodes but decrease the diameter of the graph. This is combined with a deterministic $h$-hop multi-source shortest paths scheduling technique implied by the source-detection algorithm of Lenzen and Peleg~\cite{lenzen2013efficient}, which works in the broadcast version of the \congest model. Note that other versions that could have been used, such as the one presented in~\cite{dissler-thesis,holzer2012optimal}, only work in the (multi-)unicast version. Recently Censor-Hillel and Paz \cite{computing2014censor}, as well as Kaski, Korhonen et al. \cite{kaski2014algebraisation} transferred fast matrix multiplication algorithms into the \ecc model using results from \cite{lenzen2012optimal} and derived a runtime of $\BO(n^{1/3})$ in semirings and $\BO(n^{0.15715})$ in rings. Using this they obtain an $\BO(n^{0.15715})$ algorithm for triangle detection and undirected unweighted APSP. Both papers also solve APSP on directed weighted graphs in time $\tilde\BO(n^{1/3})$. In addition \cite{kaski2014algebraisation} presents an $(1+o(1))$-approximation for exact directed weighted APSP in time $\BO(n^{0.15715})$, while \cite{computing2014censor} derives results for fast diameter and girth computation as well as for $4$-cycle detection.

\textbf{Lower bounds in the \vcc and \ecc models:} Drucker et al.~\cite{kuhn2014} were the first to provide lower bounds in the \vcc model. They derived these bounds by transferring lower bounds for set disjointness in the 3-party NOF model to the congested clique. In addition~\cite{kuhn2014} showed that explicit lower bounds in the \ecc model imply circuit lower bounds for threshold circuits (TC). While explicit lower bounds in the \ecc model remain open and might have a major impact to other fields of (Theoretical) Computer Science as mentioned above, they argue that most problems have a linear lower bound in the \ecc model by using a counting-argument. Independent and simultaneously to us, the authors of \cite{computing2014censor} presented an $\tilde\Omega(\sqrt{n})$ lower bound for APSP in the \vcc model, which they derive from matrix multiplication lower bounds that they state. 

\textbf{Lower bounds in the \congest model:} Frischknecht et al.~\cite{frischknecht2012networks} (which is based on~\cite{sarma2012distributed}) showed an $\tilde\Omega(n)$ lower bound for exact computation of the diameter of an unweighted graph. In this paper we draw on the ideas of~\cite{frischknecht2012networks} to obtain lower bounds for (2-o(1))-approximation the diameter in weighted networks. Note that also Nanongkai~\cite{nanongkai2014distributed} presents an $\tilde\Omega(n)$-time lower bound for any $poly(n)$-approximation algorithm for APSP on weighted graphs in the \congest model and shows that any $\alpha(n)$-approximation of APSP on unweighted graphs requires $\tilde\Omega(n/\alpha(n))$ time. However, his proof relies on an information-theoretic argument and uses a star-shaped graph such that it cannot be extended to the \vcc model, as in this model every node could simply broadcast its distance from the center to all other nodes. 

\textbf{Connections to systems and other models:} Finally we want to provide examples of parallel systems that might benefit from theoretical results in the \cc model. These include systems that provide all-to-all communication between $10,000$ nodes at full bandwidth~\cite{nightingale2012flat}. In addition~\cite{DBLP:conf/sirocco/HegemanP14} showed a close connection between the \ecc model and popular parallel systems such as MapReduce~\cite{dean2008mapreduce} and analyzed which kind of algorithms for the \ecc model can be simulated directly in MapReduce. Furthermore Klauk et al.~\cite{klauck2013distributed} established a connection to large-scale graph processing systems such as Pregel~\cite{malewicz2010pregel}. Finally, the authors of~\cite{kuhn2014} pointed out that the \vcc model is used in streaming~\cite{kuhn1}, cryptology~\cite{kuhn14} and mechanism design~\cite{kuhn7}. They also establish connections between the \ecc and ACC as well as TC0 circuits.

\section{Model and Definitions} \label{sec:model}

\subsection{The \congest Model}
Our network is represented by an undirected graph $G = (V, E)$, where nodes $V$ model processors or computers and edges $E$ model links between the processors. Edges can have associated \textit{weights} $w : E \rightarrow \{a/p\;|\;a\in \{1,\dots,p^2\} \subset \mathbb{N}\}$ for some $p\in poly(n)$. This ensures that each weight is a positive multiple of $1/p$ and can be encoded in $\BO(\log n)$ bits. Two nodes can communicate directly with each other if and only if they are connected by some edge from set $E$. We also assume that the nodes have unique IDs in the range of $\{1,\dots,poly(n)\}$ and infinite computational power.\footnote{This assumption is made by the model because it is used to study communication complexity. Note that we do not make use of this, as our algorithms perform efficient computations.} At the beginning, each node knows only the IDs of its neighbors and the weights of its incident edges. 

We consider a model where nodes can send messages to their neighbors over synchronous rounds of communication. During a round, each node $u$ can send a message of $B$ bits through each edge connecting $u$ to some other vertex $v$. We assume $B=\BO(\log n)$ during our algorithms, which is the standard choice~\cite{peleg} and state our lower bounds depending on arbitrary $B$. The message will arrive at node $v$ at the end of the round. We analyze the performance of an algorithm in this model by measuring the worst-case number of communication rounds required for the algorithm to complete.

\begin{definition}[Distributed Round Complexity]
Let $\mathcal{A}$ be the set of distributed deterministic algorithms that evaluate a function $g$ on the underlying graph $G$ over $n$ nodes (representing the network). Denote by $R^{dc}\left(A\left(G\right)\right)$ the distributed round complexity (indicated by dc) representing the number of rounds that an algorithm $A\in \mathcal{A}$ needs in order to compute $g\left(G\right)$. We define
	$R^{dc}\left(g\right) = \min_{A \in \mathcal{A}}\max_{G \in \mathbb{G}_n} R^{dc}\left(A\left(G\right)\right)$
to be the smallest amount of rounds/time slots any algorithm needs in order to compute $g$ on a graph $G\in \mathbb{G}_n$. Here, $\mathbb{G}_n$ is the set of all (connected) graphs over $n$ nodes. We denote by $R^{dc-pub}_\varepsilon\left(g\right)$ the (public coin\footnote{This is mainly of interest for our lower bounds. Our algorithms also work with private randomness.}) randomized round complexity of $g$ when the algorithms have access to public coin randomness and compute the desired output with an error probability smaller than $\varepsilon$.
\end{definition}

\subsection{The \cc Model}

In this model every vertex in a network $G$ can directly communicate with every other vertex in $G$. Note that although the communication graph is a clique, we are interested in solving a problem on a subgraph $G$ of the clique. Working under the \vc model and the \ec models (broadcast and (multi-)unicast versions of the \congest model), while making this assumption gives us the \vcc model and the \ecc model, respectively.

\subsection{Problems and Definitions}

For any nodes $u$ and $v\in V$, a \textit{(u,v)-path} $P$ is a path $(u = x_0, x_1, \dots, x_l = v)$ where $(x_i, x_{i+1}) \in E$ for all $i$. We define the weight of a path $P$ to be $w(P):=\sum_{i=0}^{l-1}w(x_i, x_{i+1})$. Let $P_G(u, v)$ denote the set of all (u,v)-paths in $G$. We define $d_{w}(u, v) = \min_{P \in P_G(u,v)}w(P)$; in other words, $d_{w}(u, v)$ is the weight of the shortest (weighted) path from $u$ to $v$ in $G$. The (weighted) \textit{diameter} $D_w(G)$ of $(G, w)$ is defined as $\max_{u, v\in V}\ d_{G, w}(u, v)$. For unweighted graphs $G$ (i.e. $w(e)=1$ for all $e\in E$), we omit $w$ from our notations. In particular, $d(u, v)$ is the (hop-)distance between $u$ and $v$ in $G$, and $D$ is the diameter of the unweighted network $G$.

\begin{definition}[Single Source Shortest Paths and All-Pairs Shortest Paths]
In the (weighted) single source shortest paths problem (SSSP), we are given a weighted network $(G, w)$ and a source node $s$. We want each node $v$ to know the distance $d_{w}(s, v)$ between itself and $s$. In the (weighted) all pairs shortest paths problem (APSP), each node $v\in V$ needs to know $d_{w}(u, v)$ for all $u \in V$.
\end{definition}

For any $\alpha$, we say an algorithm $A$ is an $\alpha$-\textit{approximation} algorithm for SSSP if each node $v$ obtains a value $\widetilde{d}_w(s, v)$ from $A$, such that $d_w(s, v) \leq \widetilde{d}_w(s, v) \leq \alpha \cdot d_w(s, v)$. Similarly, we say $A$ is an $\alpha$-approximation algorithm for APSP if each node $v$ obtains values $\widetilde{d}(u, v)$ such that $d_w(u, v) \leq \widetilde{d}_w(u, v) \leq \alpha d_w(u, v)$ for all $u$.

\section{Lower Bounds for Weighted and Unweighted Diameter Computation and Approximation} \label{sec:lb}
Frischknecht et al. proved in~\cite{frischknecht2012networks} that any algorithm that computes the exact diameter of an unweighted graph requires at least $\Omega(\frac{n}{B})$ rounds of communication. Note that they consider arbitrary message-size $B$, while the \congest model typically considers $B=\BO(\log n)$. We consider arbitrary $B$ as well. Their lower bound is achieved by constructing a reduction from the two-party communication problem of \textit{set disjointness} to the problem of calculating the diameter of a particular unweighted graph $G$. We extend their construction that considers exact computation of the diameter of an unweighted graphs to the case of $(2 - 1 / poly(n))$-approximating the diameter in a (positively) weighted graph. This is done by assigning weights to the edges in their (unweighted) construction in a convenient way and deriving the approximation-factor. We start by reviewing basic tools from two-party communication complexity and then present the modification of the construction of~\cite{frischknecht2012networks} for the \congest model in Section~\ref{2approx}. Section~\ref{sec:2-approx-vcc} transfers this bound to the \vcc model. 

\subsection{A Review of Basic Two-Party Communication Complexity}
It is necessary to review the basics of two-party communication complexity in order to present our results in a self-contained way. In the remaining part of this subsection we restate the presentation given in~\cite{holzer2013phd} only for completeness and convenience of the reader.

Two computationally unbounded parties Alice and Bob each receive a $k$-bit string $a\in\{0,1\}^k$ and $b\in\{0,1\}^k$ respectively. Alice and Bob can communicate with each other one bit at a time and want to evaluate a function $h : \{0,1\}^k \times \{0,1\}^k \rightarrow \{0,1\}$ on their input. 
We assume that Alice and Bob have access to public randomness for their computation and we are interested in the number of bits that Alice and Bob need to exchange in order to compute $h$. 
\begin{definition}[Communication complexity]
Let $\mathcal{A}_\delta$ be the set of two-party algorithms that use public randomness (denoted by pub), which when used by Alice and Bob, compute $h$ on any input $a$ (to Alice) and $b$ (to Bob) with an error probability smaller than $\delta$. Let $A\in \mathcal{A_\delta}$ be an algorithm that computes $h$. Denote by $R^{cc-pub}_\delta(A(a,b))$ the communication complexity (denoted by cc) representing the number of $1$-bit messages exchanged by Alice and Bob while executing algorithm $A$ on $a$ and $b$. We define
	$$R^{cc-pub}_\delta(h) = \min_{A \in \mathcal{A_\delta}}\max_{a,b\in\{0,1\}^k} R^{cc-pub}(A(a,b))$$
to be the smallest amount of bits any algorithm would need to send in order to compute $h$.
\end{definition}
A well-studied problem in communication complexity is that of set disjointness, where we are given two subsets of $\{0,\dots,k-1\}$ and need to decide whether they are disjoint. Here, the strings $a$ and $b$ indicate membership of elements to each of these sets.

\begin{definition}[Disjointness problem]\label{appendix:def:disj}
The set disjointness function $\disj{k}: {\{0,1\}}^k \times {\{0,1\}}^k \rightarrow {\{0,1\}}$ is defined as follows.
	$$\disj{k}(a,b) =
	\begin{cases}
		0 &: \text{if there is an $i\in \{0,\dots,k-1\}$ such that $a(i)=b(i)=1$} \\
		1 &: \text{otherwise}
	\end{cases}$$
	where $a(i)$ and $b(i)$ are the $i$-th bit of $a$ and $b$ respectively (indicating whether an element is a member of the corresponding set.)
\end{definition}
We use the following basic theorem that was proven in Example 3.22 in \cite{kushilevitz97} and in \cite{BabaiFS86,Bar-YossefJKS04,KalyanasundaramS92,Razborov92}. 
\begin{theorem}\label{appendix:thm:disj}
For any sufficiently small $\delta>0$ we can bound $R^{cc-pub}_\delta(\disj{k})$ by $\Omega(k)$.
\end{theorem}

\subsection{Lower Bounds for Weighted Diameter Computation in the \congest Model}
\begin{theorem} \label{2approx}
For any $n \geq 10$ and $B \geq 1$ and sufficiently small $\varepsilon$ any distributed randomized $\varepsilon$-error algorithm $A$ that computes a $(2 - 1 / poly(n))$-approximation of the diameter of a positively weighted graph requires at least $\Omega(\frac{n}{B})$ time for some $n$-node graph.
\end{theorem}

We follow the strategy of~\cite{frischknecht2012networks} and reduce the function $disj_{k(n)^2}$ to finding the diameter of an graph $G$. Note that the graph in~\cite{frischknecht2012networks} is unweighted, while ours is weighted. We set a parameter $k(n)$ to be $k(n) = \lfloor \frac{n}{10} \rfloor$ and construct a graph $G_{a, b}$. We do so by defining a graph $G_a=(V_a,E_a)$ that depend on inputs $a$ and a graph $G_b=(V_b,E_b)$ that depends on $b$. Based on these sets we derive graph $G_{a,b}$. We start by construct sets of nodes
$L = \{l_v | v \in \{1, \dots, 2k(n)-1\}\}\text{ and }R = \{r_v | v \in \{1, \dots, 2k(n)-1\}\}.
$
\ Let
$
L_1 = \{l_v | v \in \{1, \dots, k(n)-1\}\}\text{ and }L_2 = \{l_v | v \in \{k(n), \dots, 2k(n)-1\}\}, 
$
\ and define 
$R_1 = \{r_v | v \in \{1, \dots, k(n)-1\}\}\text{ and }R_2 = \{r_v | v \in \{k(n), \dots, 2k(n)-1\}\}.$ We add a node $c_L$ to $V_a$ and a node $c_R$ to $V_b$, then add edges from $c_L$ to all nodes in $L$ and from $c_R$ to all nodes in $R$. We also add edges between each pair of nodes in $L_1$, $R_1$, $L_2$, and $R_2$, and from $l_i$ to $r_i$ for $i \in \{1, \dots, 2k(n)-1\}$. Finally, we add an edge from $c_L$ to $c_R$. Note that these sets of (right/left) nodes only depend on the lengths of the inputs. In the proof we define edges $E_a$ that connect nodes in $V_a$ depending on $a$. We also define edges $E_b$ that connect nodes in $V_b$ depending on $b$.

\begin{proof}

As in~\cite{frischknecht2012networks}, we can represent the $k(n)^2 - 1$ bits of input $a$ by the $k(n)^2$ possible edges between the $k(n)$ nodes $L_1$ and $k(n)$ nodes $L_2$. More specifically, we choose the mapping from integers in $\{1, \dots, k(n)^2 - 1\}$ to pairs of integers in $\{1, \dots, k(n)-1\} \times \{k(n), \dots, 2k(n)-1\}$, such that $i$ is mapped to $(l_{u_i}, l_{v_i}) = \left(i \mod k(n), k(n) + \left\lfloor \frac{i}{k(n)}\right\rfloor\right)$. We add edge $(l_{u_i}, l_{v_i})$ to $G_a$ if and only if $a(i) = 0$, and likewise represent the bits of $b$ by adding edge $(r_{u_i}, r_{v_i})$ to $G_b$ if and only if $b(i) = 0$.

We call the graph defined by these edges $G_a = (V_a, E_a)$, and construct a similar graph $G_b$ for input $b$. We define the cut-set $C_{k(n)^2} = \{(l_v, r_v) : v \in \{0, \dots, 2k(n)-1\}\}$ to be the $2k(n)$ edges connecting each $l_v$ to the corresponding $r_v$. We will refer to the sets of vertices
$L_1\cup R_1=\{l_v | v \in \{1, \dots, k(n)-1\}\} \cup \{r_v | v \in \{1, \dots, k(n)-1\}\}\text{ as }\textbf{UP }\text{(upper part of the graph)}$ 
and
$L_2\cup R_2=\{l_v | v \in \{k(n), \dots, 2k(n)-1\}\} \cup \{r_v | v \in \{k(n), \dots, 2k(n)-1\}\}\text{ as }\textbf{LP}$ (lower part of the graph). In the figure below, we note that the former is in the upper portion of the graph, and the latter is in the lower portion.
Finally, we set $G_{a,b} = G_a\cup G_b\cup C_k$.

Now we assign weights to the edges in this construction. We set the weight of every edge in $G_a$ and in $G_b$ to be $1$, and the weight of each edge in $C_{k(n)^2}$ to be $1/p$, the smallest possible weight (see definition of the weights in Section~\ref{sec:model}). 

\begin{lemma} \label{diam2}
The weighted diameter of $G_{a, b}$ is at most $2+1/p$.
\end{lemma}
\begin{proof}
This proof can be found in Appendix~\ref{app:diam2}.
\end{proof}

Following the ideas of~\cite{frischknecht2012networks}, we reduce the problem of deciding disjointness between sets $a$ and $b$ to computing the diameter of a graph.

\begin{figure}[ht!]
	\begin{center}
		\includegraphics[width=.6\textwidth]{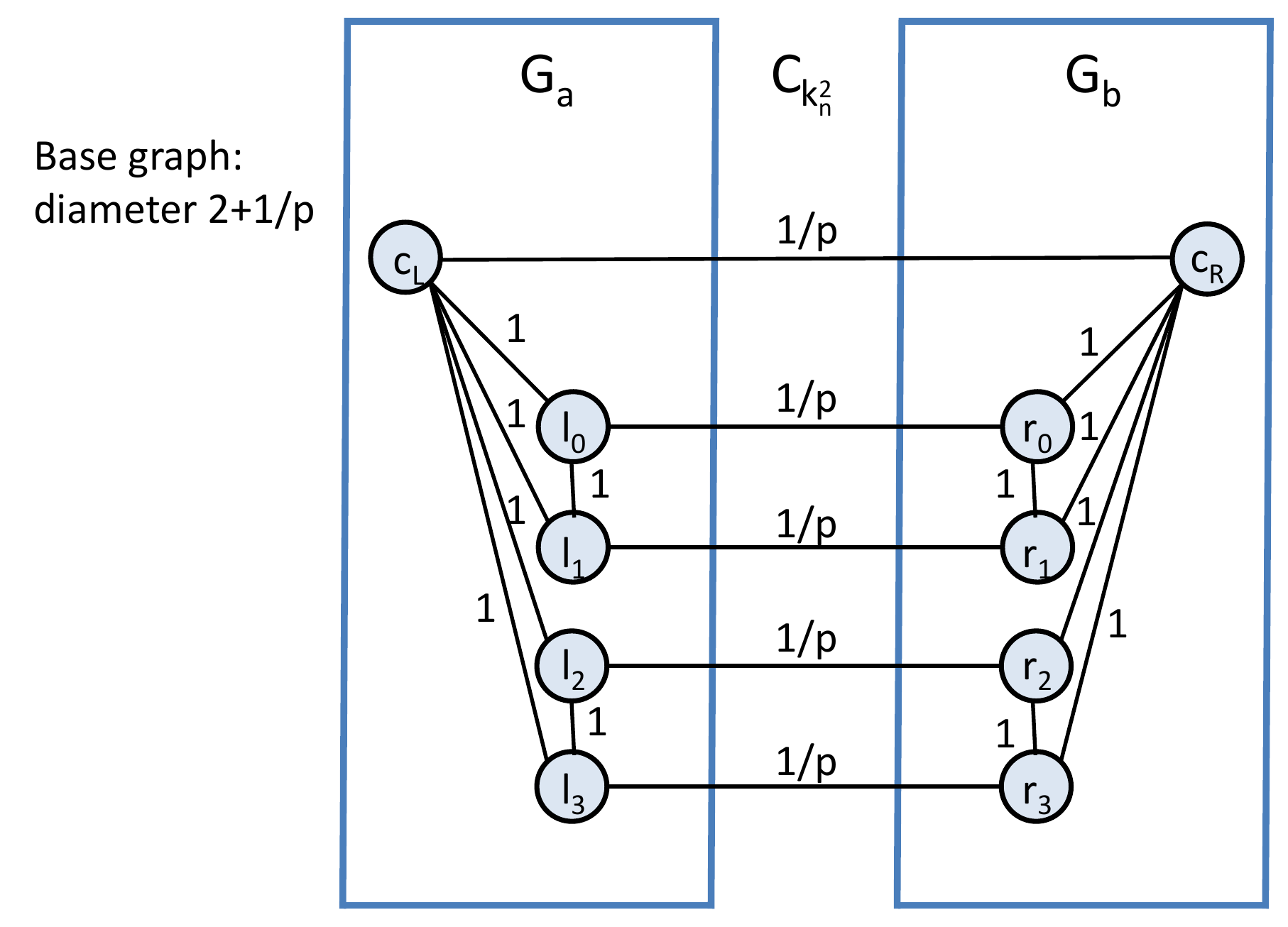}
	\end{center}
	\caption{Base graph of (weighted) diameter $2+1/p$.}
\end{figure}

\begin{figure}[ht!]
	\begin{center}
		\includegraphics[width=.6\textwidth]{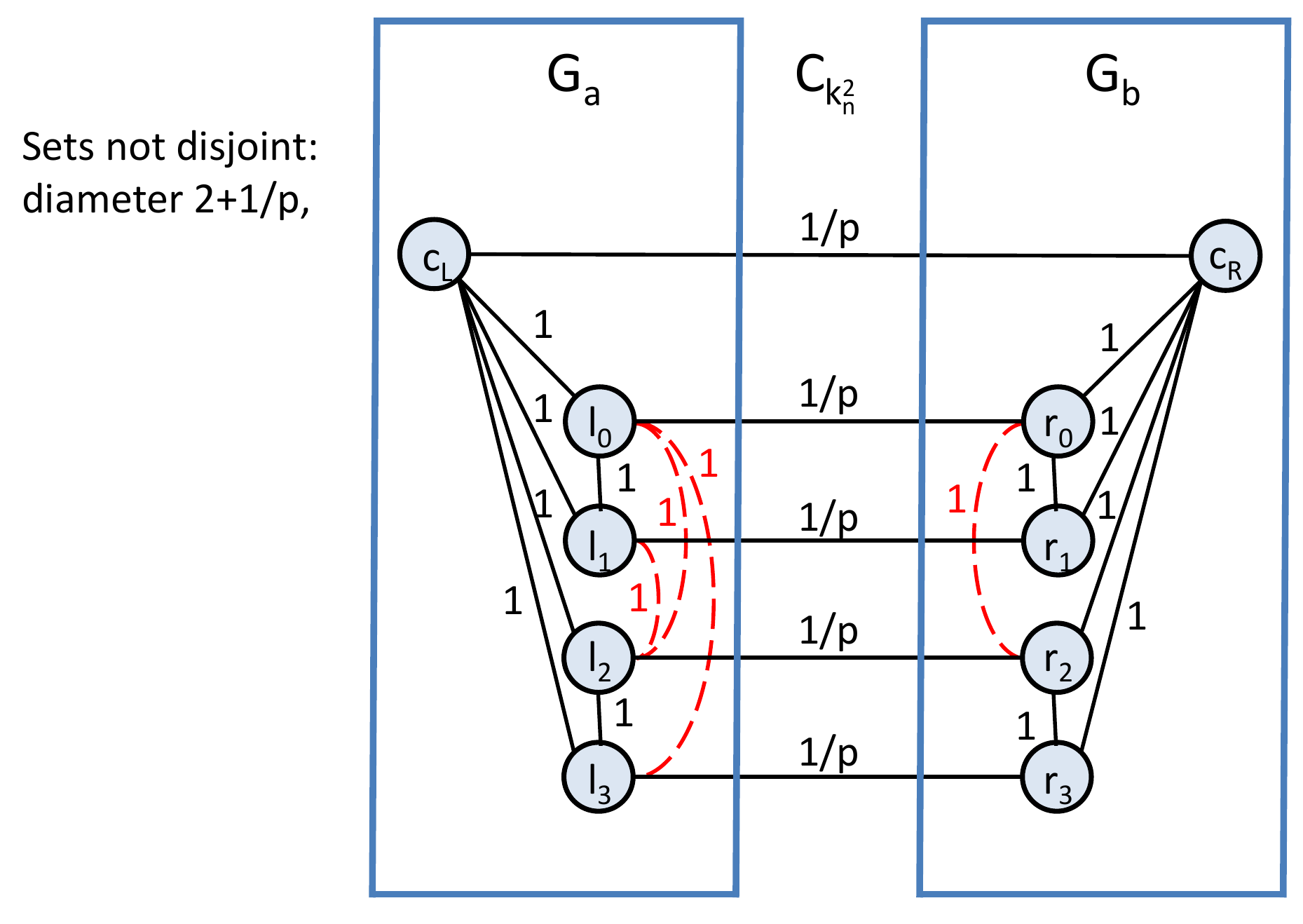}
\\ 
		\includegraphics[width=.6\textwidth]{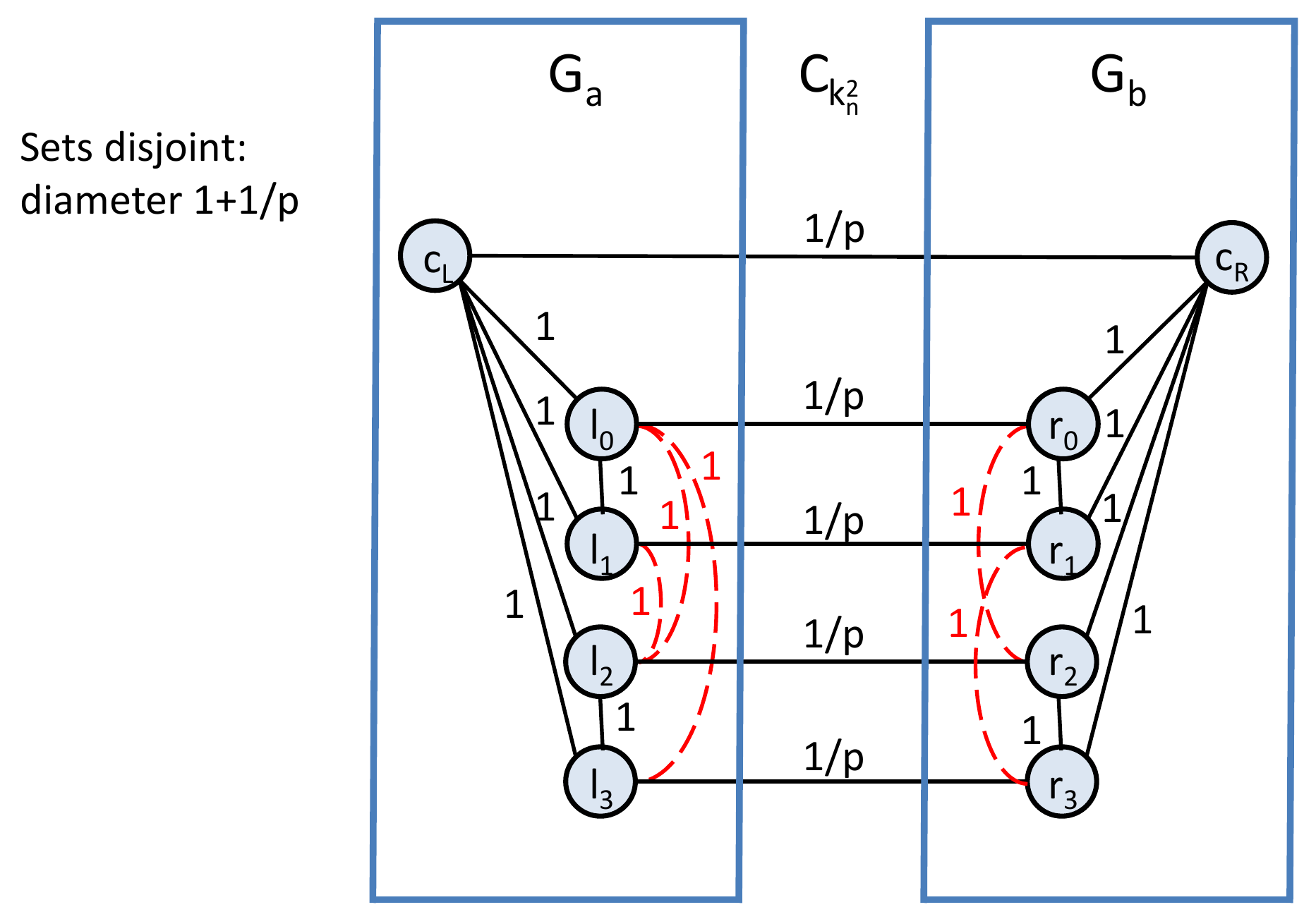}
	\end{center}
	\caption{Top: Input sets $a$ and $b$ are not disjoint: the index $i$ corresponding to the pair $(1, 3)$ has $a(i) = b(i) = 1$, s.t.  the diameter is $2+1/p$. Bottom: Input sets $a$ and $b$ are disjoint: every pair of integers $(i, j) \in \{0, \dots, k(n)-1\} \times \{k(n), \cdots, 2k(n)-1\}$ has either $(l_i, l_j) \in G$ or $(r_i, r_j) \in G$, so the diameter is $1+1/p$.}
\end{figure}

\begin{lemma} \label{2approx2}
The diameter of $G_{a, b}$ is 1 if the sets $a$ and $b$ are disjoint, else it is $2$.
\end{lemma}
\begin{proof}
This proof can be found in Appendix~\ref{app:2approx2}.
\end{proof}

We are now equipped to prove Theorem \ref{2approx} above. We use the graph $G_{a, b}$ constructed above to show that any algorithm $A$ that computes a $(2-1/p)$-approximation of the diameter requires $\Theta(\frac{n}{B})$ time. 

First note, that in case the diameter is $(1+1/p)$ any $A$ must output a value of at most $(1+1/p)(2-1/p)= 2+1/p-1-p^2$. As this value is strictly smaller than the other possible diameter of $G_{a,b}$, which is ($2+1/p$), any $(2-1/p)$-approximation algorithm can decide whether the $D_w(G_{a, b})$ is $(1+1/p)$ or $(2+1/p)$. Based on this one can decide if inputs $a$ and $b$, that were used to construct the graph $G_{a,b}$, are disjoint.

However, we know due to Theorem~\ref{appendix:thm:disj} that any algorithm must exchange $\Omega(k(n)^2)$ bits of information through the edges in $C_{k(n)}$ in order to decide if $a$ and $b$ are disjoint. As the bandwidth of $C_{k(n)}$ is $\BO(|C_{k(n)}|\cdot B)=\BO(k(n)\cdot B)$, we conclude that $\Omega(k(n)/B)$ rounds are necessary to do so.
Due to the choice of $k(n)$ we conclude that $\Omega(\frac{n}{B})$ rounds are necessary to $(2-1/p)$-approximate the diameter of a graph. 
\end{proof}

\subsection{Lower Bounds for Weighted Diameter Computation in the \vcc Model}\label{sec:2-approx-vcc}

\begin{theorem} \label{thm:2-party-lb}
Given a two-party communication problem $f'$ that can be reduced to a graph $G_{a,b}$ and a randomized algorithm $A$ in the \vcc model, if $R_\epsilon^{cc-pub}(f')$ is a lower bound on the number of bits that must be communicated in $f'$, then $A$ must take at least $\frac{R_\epsilon^{cc-pub}(f')}{nB}$ rounds in the \vcc model.
\end{theorem}
Before starting with the proof, we want to stress that that the edges of $G_{a,b}$ remain the only edges with weights. Other edges of the clique not mentioned in the construction of $G_{a,b}$ are only present in the \cc model (not in the \congest model studied in Section~\ref{2approx}) and are only used for communication. These (additional) communication edges are assigned no weight, as they are not part of the lower bound construction and do not affect the diameter of the graph $G_{a,b}$.

\begin{proof}

In each round, any algorithm can send at most $|G_a| \cdot B$ bits of information from $G_a$ to $G_b$, as each vertex in $G_a$ must broadcast the same $B$ bits to all other vertices in $G_b$ in the $\vcc$ model. Similarly, any algorithm can send at most $|G_b| \cdot B$ bits from $G_b$ to $G_a$. There are no further nodes outside of $G_{a,b}$ that could increase the bandwidth. Thus, any algorithm can exchange at most $(|G_a| + |G_b|) \cdot B = nB$ bits between $G_a$ and $G_b$ in each round. Therefore $\frac{R_\epsilon^{cc-pub}(f')}{nB}$ is a lower bound on the number of rounds that algorithm $A$ must take.
\end{proof}

\begin{theorem}\label{thm:2-approxLB-vcc}
Computing a $(2 - o(1))$-approximation of the diameter in positively weighted graphs in the \vcc model takes $\Omega(n/B)$ rounds.
\end{theorem}
\begin{proof}
Computing a $(2 - o(1))$-approximation of the diameter in positively weighted graphs is shown to require the exchange of $\Omega(n^2)$ bits of information, by Theorem \ref{2approx} above. The statement then follows directly from an application of Theorem \ref{thm:2-party-lb}.
\end{proof}
\begin{theorem}\label{thm:exactLB-vcc}
Computing the diameter exactly in unweighted graphs takes $\Omega(n/B)$ in the \vcc model.
\end{theorem}
\begin{proof}
Computing the exact diameter of unweighted version of the graph $G_{a,b}$ is shown to require $\Omega(n^2/B)$ bits of information to be exchanged in~\cite{frischknecht2012networks}. Thus, the result follows by Theorem \ref{thm:2-party-lb} using similar arguments as in the proof of Theorem~\ref{thm:2-approxLB-vcc}.
\end{proof}
\begin{remark}\label{rem:2diam}
Note that a $2$-approximation of the diameter of positively weighted graphs is achievable by computing SSSP starting in an arbitrary node, and returning twice the length of the largest distance computed. To compute (exact) SSSP-algorithm we can use the SSSP-algorithm presented in~\cite{nanongkai2014distributed}, that runs in $\tilde\BO(\sqrt{n})$ time.
\end{remark} 

\section{Deterministic Hitting Set Computation in the \vcc Model.} \label{sec:vccHS}

\begin{definition}
Given a node $u \in V$, the set $S^k(u)$ of a node $u \in G$ contains the $k$ nodes closest to $u$ in a weighted graph $G$, with ties broken by node ID. In other words, $S^k(u) \subset V$ has the following properties:
\begin{enumerate}
\item $|S^k(u)| = k$, and
\item for all $s \in S^k(u)$ and $t \notin S^k(u)$, either (i) $d_w(u, s) < d_w(u, t)$, or (ii) $d_w(u, s) = d_w(u, t)$ and the ID of $s$ is smaller than the ID of $t$.
\end{enumerate}
\end{definition}

\begin{definition}
A $k$-hitting set $S$ of a graph $G = (V, E)$ is a set of nodes such that, for every node $v \in V$, there is at least one node of $S$ in $S^k(v)$.
\end{definition}

Algorithm~\ref{alg:HS} takes as input a graph $G$ and an integer $k$, and returns a $k$-hitting set $S \subseteq V$. The algorithm works as follows: each node starts by broadcasting its $k$ incident edges of smallest weight to all other nodes (Lines \ref{hsl1}-\ref{hsl2}). If the node does have less than $k$ neighbors, it just broadcasts the weight of all its incident edges. This enables every node $u$ to locally compute a set $S^k(u)$ (Line \ref{hsl22}), consisting of the $k$ closest nodes to $u$ in $G$ (\cite{nanongkai2014distributed}, Observation 3.12). By closest we refer to the distance of nodes to $u$ and remark that $S^k(u)$ always has $k$ nodes for any $k\leq n$, as the graph is connected. We initialize $S := \emptyset$; $S$ is updated over time until it is our desired $k$-hitting set. Let at any time $R$ be composed of the sets $S^k(v)$ such that $S^k(v)\ \cap\ S = \emptyset$ (initially $R$ contains all $S^k(v)$). We repeatedly find the vertex $v_{max}$ that is contained in the largest number of elements in $R$ (breaking ties by minimum node ID). We then add this $v_{max}$ to $S$ and update $R$ accordingly. In Lemma~\ref{greedyvc} we show that this method of greedily constructing a hitting set achieves a $\BO(\log n)$-approximation of the smallest possible hitting set. 

\begin{lemma}\label{greedyvc}
Given a graph $G$, if the smallest possible hitting set uses $N$ vertices, then $S$ contains at most $\BO(N \log n)$ vertices.
\end{lemma}
\begin{proof}
This proof can be found in Appendix \ref{app:greedyvc}.
\end{proof}

\begin{algorithm}
\caption{$\tilde\BO(n/k)$-time deterministic $(k)$-hitting set algorithm}\label{alg:HS}
\begin{algorithmic}[1]

\Procedure{HittingSet}{$G, k$}
\Comment{as executed by each node $u \in G$}

\For{$i \in \{1, \dots, k\}$} \label{hsl1}
	\State Broadcast the $i$-th lowest weight adjacent edge to all nodes.
\EndFor \label{hsl2}

\State Locally compute $S^k(u)$. \label{hsl22}

\For{$i \in \{1, \dots, k\}$} \label{hsl11}
	\State \parbox[t]{\dimexpr\linewidth-\algorithmicindent}{Broadcast the $i$-th lowest weight edge in $S^k(u)$ to all nodes, and add the received edge from each node $v$ to a set $S_v$.\strut}
\EndFor \label{hsl112}

\State $S \gets \emptyset$ \label{hsl3}
\State $R \gets V$
\While{$\exists S_i : S_i \cap S = \emptyset$}
	\State $w \gets NULL$
	\State $R \gets G\ \backslash\ S$
	\ForAll{nodes $v \in R$}
		\State $N_v \gets \{S_i : v \in S_i\text{ and } S_i \cap S \neq \emptyset\}$
		\If{$w = NULL$ or $|N_v| > |N_w|$}
			\State $w \gets v$
		\EndIf
	\EndFor
	\State $S \gets S \cup \{w\}$
\EndWhile \label{hsl4}

\State \textbf{return} S

\EndProcedure

\end{algorithmic}
\end{algorithm}

\begin{lemma} \label{lem:HS-runtime}
Procedure \textsc{HittingSet} described in Algorithm~\ref{alg:HS} computes a $k$-hitting set of size $\tilde{\BO}(n/k)$ in $\BO(k)$ rounds.\footnote{By using the $\BO$-notation we implicitly assume that $k\leq n^{1-\text{polylog} n}$, which will always be the case in this paper.}
\end{lemma}
\begin{proof}
\textit{Runtime:} Each node begins by broadcasting its $k$ minimum-weight edges to all other nodes (see Lines \ref{hsl1}-\ref{hsl2} of Algorithm~\ref{alg:HS}), which takes $\BO(k)$ rounds. Lines \ref{hsl11}-\ref{hsl112} consist of $k$ repetitions of broadcasting a single edge, and thus also take $\BO(k)$ rounds. Line \ref{hsl22} and Lines \ref{hsl3}-\ref{hsl4} only consist of local computation, and can be completed without any additional communication (and thus need no round of communication). Thus, the total number of rounds required is $\BO(k)$.

\textit{Size of the $k$-hitting set:} A random subset of the nodes of size $\tilde\BO(n/k)$ is a $k$-hitting set with high probability (see~\cite{nanongkai2014distributed}). Therefore, the minimum number $N$ of nodes in a $k$-hitting set is upper bounded by $\tilde\BO(n/k)$, and we apply Lemma \ref{greedyvc} to conclude that the set $S$ that is computed by Algorithm~\ref{alg:HS} contains $\tilde\BO(n/k)$ nodes, as desired.

\end{proof}

\section{Deterministic (2 + o(1))-Approximation of APSP in Time 
$\tilde\BO(\text{n}^{1/2})$ in the \vcc Model} \label{sec:vccAPSP}

Nanongkai provides a randomized distributed algorithm (\cite{nanongkai2014distributed}, Algorithm 5.2) to $(2 + o(1))$-approximate APSP in the \vcc model that runs in $\tilde\BO(n^{1/2})$ time. At a high level, this algorithm works by
\begin{enumerate}
\item choosing a random \textit{$\sqrt{n}$-hitting set} $R \subseteq V$ of size $\tilde\BO(\sqrt{n})$such that for all nodes in $V$, there is some node in $R$ within $\sqrt{n}$ hops,
\item $(1+o(1))$-approximate (using random delays to avoid congestion) shortest paths from each node in the hitting set $R$ to every node in $V$,
\item using these shortest paths to approximate shortest paths between all pairs of nodes.
\end{enumerate}

We already presented a method to deterministically compute a $\sqrt{n}$-hitting set $R \subseteq V$ in Section \ref{sec:vccHS}. In the second part of his algorithm, Nanongkai uses a randomized procedure as well, which we replace by a deterministic one in this paper. This results in a deterministic $\tilde{\BO}(n^{1/2})$ round algorithm and we state:

\begin{theorem}\label{thm:APSP-vcc}
The deterministic Algorithm \ref{alg:APSP-vcc} (stated below) returns a $(2 + o(1))$-approximation of APSP in time $\tilde\BO(n^{1/2})$.
\end{theorem}
The remainder of this section is devoted to explaining and analyzing Algorithm \ref{alg:APSP-vcc}, which proves this theorem in the end. While doing so, we also review the whole Algorithm 5.2 of \cite{nanongkaiarxiv}. We do this to be able to point out our modifications exactly and to argue that each step can indeed be done in the \vcc model, while the original implementation of Algorithm 5.2 of \cite{nanongkaiarxiv} is just stated for the \cc model (without distinguishing between \vcc and \ecc models). As shown in Theorem 5.3 of \cite{nanongkaiarxiv}, Algorithm 5.2 of \cite{nanongkaiarxiv} computes a $(2 + o(1))$-approximation of APSP on weighted graphs. Note that we only change the implementation of Algorithm 5.2 of \cite{nanongkaiarxiv} to be deterministic such that we can derive the same approximation ratio (with probability one instead of w.h.p.).  

Given a graph $G$, Nanongkai \cite{nanongkaiarxiv} starts by computing a \textit{$k$-shortcut graph} $G^k$ of $G$ for $k = \sqrt{n}$.
\begin{definition}[$k$-shortcut graph]
The shortcut graph $G^k = (V, E^k)$ is obtained by adding an edge $(u,v)$ of weight $d_{w}(u, v)$ to $E^k$ for every $u \in V$ and $v \in S^k(u)$.
\end{definition}
To construct this graph (Lines \ref{APSP-vcc-setk}--\ref{apsp3}), each node begins by broadcasting the $k$ lightest edges adjacent to it. If there are less than $k$ edges adjacent to a node, that node just broadcasts all of them and their weights. Based on this information each node $u\in V$ can compute $S^k(u)$, since running e.g.
$k$ rounds of Dijkstra's algorithm will only need the $k$-lightest edges incident to each node (as argued in \cite{nanongkaiarxiv}). 
During the next $O(k)$ time steps, each node $u$ simultaneously broadcasts its $S^k(u)$ and creates a simulated shortcut edge from every node $u \in G$ to every node $v \in S^k(u)$. New edge weights $w'(u, v) := \min\{w(u, v), \min_{z \in S^k(u)} d_w(u, z) + d_w(z, v)\}$ are assigned to this graph (Lines \ref{apsp41}--\ref{apsp42}). 
Then, in Line \ref{apsp43}, node $u$ locally computes a $k$-hitting set $R$ of $G$, as described in Section \ref{sec:vccHS}, Algorithm \ref{alg:HS}.

\begin{algorithm}[ht!] 
\caption{Deterministic $\tilde\BO(n^{1/2})$-time $(2 + o(1))$-approximation algorithm for APSP in the \vcc model}\label{alg:APSP-vcc}
\begin{algorithmic}[1]

\Procedure{APSP-\vcc}{$G, w$}
\Comment{as executed by each node $u \in G$}

\State $k \gets n^{1/2}$ \label{APSP-vcc-setk}

\For{i in \{1, \dots, k\}} \label{apsp1}
	\State Broadcast the $i$-th lowest weight adjacent edges to all nodes. 
\EndFor \label{apsp2}

\State Compute and broadcast $S^k(u)$ and $\{d_w(u, z)\}_{z \in S^k(u)}$. \label{apsp3}

\ForAll{nodes $v \in V$} \label{apsp41}
	\State $w'(u, v) \gets \min\{w(u, v), \min_{z \in S^k(u)} d_w(u, z) + d_w(z, v)\}$
\EndFor \label{apsp42}

\State $R \gets \textsc{HittingSet}(G, k)$\label{apsp43}

\State $\epsilon \gets \frac{1}{\log n}$ \label{mssp-ecc-line1}
\State $h \gets 4n^{1/2}$
\State $\epsilon\gets 1/\log n$ \label{mssp:eps}
\State $W \gets \max_{e \in E}w'(e)$

\ForAll{$i \in [0, \log W]$}
	\State $D_i' \gets 2^i$
	\State $w_i'(x, y) = \left\lceil \frac{2hw'(x, y)}{\epsilon D_i'} \right\rceil$
\EndFor \label{mssp-ecc-line2}

\ForAll{$i \in [0, \log W]$} \label{mssp-ecc-line3}
	\State Transform each weighted edge $(x, y)$ into $w_i'(x, y)\leq W$ unweighted edges.
	\State \parbox[t]{\dimexpr\linewidth-\algorithmicindent}{Run $(R,h,|R|)$-source detection algorithm (Lemma \ref{lem:sources}) using the obtained unweighted graph for $h+|R|+1$ time steps.}
	\State $d_i'(R, u) \gets $ the distance returned to node $u$, or $\infty$ if no distance was returned.
\EndFor \label{mssp-ecc-line4}

\ForAll{$s_i \in R$}\label{line:3}
	\State $\widetilde{d}^h_{w}(s_i, u) \gets \min_{j \in [0, \log W]}\ d_j'(s_i, u)$
\EndFor\label{line:4}

\ForAll{nodes $v \in R$}\label{line:1}
	\State Broadcast $(v, d'(u, v))$.
\EndFor\label{line:2}

\State $d''(u, v) = \min_{r \in R} d'(u, r) + d'(r, v)$
\EndProcedure

\end{algorithmic}
\end{algorithm}

To further describe the algorithm we need the following definitions.
\begin{definition}[$h$-hop SSSP (\cite{nanongkai2014distributed}, Definition 3.1)]
Consider a network $(G, w)$ and a given integer $h$. For any nodes $u$ and $v$, let $P^h(u, v)$ be the set of all $(u, v)$-paths containing at most $h$ edges. Define the $h$-hop distance between $u$ and $v$ as
\begin{equation*}
    d_{w}^h(u, v) = \begin{cases}
               min_{P \in P^h(u, v)} w(P) 	& :\ P^h(u, v) \neq \emptyset \\
               \infty 				&: \ otherwise.
           \end{cases}
\end{equation*}
Let $h$-hop
SSSP be the problem where, for a given weighted network $(G,w)$, source node
$s$ (node $s$ knows that it is the source), and integer $h$ (known to every node), we want every node $u$
to know $dist^h_{G,w}(s,u)$.

\end{definition}
\begin{definition}[MSSP, $h$-hop MSSP \cite{nanongkaiarxiv} (a.k.a. ($h$-hop) $S$-SP \cite{dissler-thesis,holzer2012optimal})]
Given a set $S\subseteq V$, the multi-source shortest paths problem (MSSP) (a.k.a. $S$-shortest paths problem ($S$-SP)) is to compute SSSP from each node in $S$. In the $h$-hop MSSP problem (a.k.a. $h$-hop $S$-SP) one is interested in the $h$-hop versions of SSSP w.r.t source nodes $S$.
\end{definition}

Nanongkai states an MSSP algorithm that works in the \congest model, and computes $(1 + o(1))$-approximate distances on weighted graphs. The main idea of this algorithm is based on the following theorem.

\begin{theorem}[\cite{nanongkai2014distributed}, Theorem 3.3] \label{thm:SSSPapprox}
Consider any $n$-node weighted graph $(G, w)$ and integer $h$. Let $\epsilon = 1 / \log n$, and let $W$ be the maximum-weight edge in $G$. For any $i$ and edge $(x, y)$, let $D_i' = 2^i$ and $w_i'(x, y) = \left\lceil \frac{2hw(x, y)}{\epsilon D_i'} \right\rceil$. For any nodes $u$ and $v$, if we let
\[
\widetilde{d}_{w}^h(u, v) = min\ \bigg\{ \frac{\epsilon D_i'}{2h} \times d_{w_i'}(u, v)\ |\ i : d_{w_i'}(u, v) \leq (1 + 2/\epsilon)h \bigg\},
\]
then $d_{w}^h(u, v) \leq \widetilde{d}_{w}^h(u, v) \leq (1 + \epsilon)\cdot d_{w}^h(u, v)$.
\end{theorem}
This theorem states that we can  compute an $(1 + \varepsilon)$-approximation of $h$-hop-bounded SSSP when we run $\BO(\log n)$ many $h$-hop-bounded SSSP computations rooted in node $u$, each with modified weights $w_1',\dots.w_{\log n}'$. To obtain an $(1 + \varepsilon)$-approximation for $h$-hop-bounded MSSP for sources $S$, Nanongkai performs $\BO(\log n)$ many $h$-hop-bounded MSSP computations rooted in $S$, each with modified weights $w_1',\dots.w_{\log n}'$. In each execution of a $h$-hop MSSP, Nanongkai starts all $h$-hop SSSP computations in all nodes of $S$ simultaneously and delays each step of any $h$-hop SSSP algorithm by a random amount. This is shown to guarantee that with high probability the $|S|$ copies of $h$-hop SSSP do not conflict with each other. 

We can adapt Nanongkai's $h$-hop MSSP algorithm to a deterministic setting using the source detection algorithm of \cite{lenzen2013efficient}.

\begin{definition}[$(S,H,K)$-source detection \cite{lenzen2013efficient}]
Given an unweighted graph $G$ and $H,K\in \mathbb{N}_0$, the $(S,H,K)$-source detection problem is to output for each node $u\in V$ the set $L_u(H,K)$ of all (up to) $K$ closest sources in $S$ to $u$, which are at most $H$ hops away.
\end{definition}
\begin{lemma}[Theorem 4.4, \cite{lenzen2013efficient}]\label{lem:sources}
The $(S,H,K)$-source detection problem can be solved in the \congest model in $\min (H,D) + \min(K,|S|)$ rounds.
\end{lemma}
In Algorithm 1 of \cite{lenzen2013efficient} that corresponds to Lemma \ref{lem:sources}, each node always broadcasts the same message within each time step to all neighbors. Therefore it runs in the broadcast version of the \congest model. Furthermore, it implicitly computes (bounds on) distances that it uses to figure out which nodes are the $K$ closest ones. In the end, these bounds correspond to the exact distances for the $K$ closest nodes. 

Now we proceed by adapting this algorithm that is stated for unweighted graphs to weighted graphs by replacing every edge $e$ of weight $w(e)$ by a path of $w(e)$ edges, each of weight $1$. The simulation of these new nodes and edges is handled by the two nodes adjacent to $e$, and is equivalent to delaying any transmission through $e$ by $w(e)$ rounds as it is done in~\cite{nanongkai2014distributed}. This transforms a weighted graph into an unweighted one.

We now use the above deterministic procedure instead of Nanongkai's randomized one to approximate weighted $h$-hop MSSP on the hitting set. That is we choose $S:=R$. In each execution of the unweighted $h$-hop MSSP on $R$, during iteration $i$, set the weight $w_i'(x, y)$ to be $\left\lceil \frac{2hw'(x, y)}{\epsilon 2^i} \right\rceil$. Now execute Lenzen and Peleg's $(S,H,K)$-source detection algorithm (Lemma \ref{lem:sources}) on graph $G^k$ using weight $w_i'$ with $R:=S$ and $H:=h$. Furthermore we set $K:=|R|$ to guarantee that all sources within $h$ hops are detected. Here we use the fact that in our model nodes at any distance in the graph $G$ can directly communicate with each other. Therefore the runtime of the algorithm stated for the \congest model applies to $G^k$ as well (and not only to $G$) in the \vcc model.

After all $\BO(\log n)$ executions have completed, each node $u \in V$ knows its distance to every node in $R$ under every set of weights $w_i$. 
By Theorem \ref{thm:SSSPapprox}, this allows us to compute a $(1 + o(1))$-approximation of $d_{w}^h(s, u)$ on $G^k$ (Lines \ref{line:3}--\ref{line:4}) when choosing $\varepsilon = 1 / \log n$ (in Line \ref{mssp:eps}), which according to \cite{nanongkaiarxiv} is equal to $d_{w}(s, u)$ for each $s\in R$ and $u\in V$. This is proven by in \cite{nanongkaiarxiv} via the choice of $h$ and $k$, which we do not change. Finally we broadcast these weights in Lines \ref{line:1}--\ref{line:2} and compute like in \cite{nanongkaiarxiv} the value $d''(u,v)$, which Nanongkai bounds to be a $(2+o(1))$-approximation.

\begin{proof}[of Theorem \ref{thm:APSP-vcc}]
\textit{Runtime}: Broadcasting the $k$ lowest-weight edges, one by one in each round, takes $k$ rounds in the \vcc model. Computing $S^k(u)$and $w'$ takes no additional communication. By Lemma \ref{greedyvc} we can compute the $k$-hitting set $R$ is computed in time $\BO(k)$ in the \vcc model.  
Computing weights $w_i'$ in Lines \ref{mssp-ecc-line1}-\ref{mssp-ecc-line2} takes $\BO(\log W)$ rounds. Lines \ref{mssp-ecc-line3}-\ref{mssp-ecc-line4} take $\BO(\log W)$ iterations, each of $\BO(h + |R|)$ time, as each execution of $(R,h,|R|)$-source detection takes $h+|R|$ time steps on the (simulated) undirected graph, see Lemma \ref{lem:sources}. Since $h = \BO(n^{1/2})$ and $|R|=\tilde\BO(n/k)=\tilde\BO(\sqrt{n})$ (see Lemma \ref{greedyvc}) and $\log W=\BO(\log n)$, as $W\in$ poly $n$, Lines \ref{mssp-ecc-line3}-\ref{mssp-ecc-line4} take $\tilde\BO(n^{1/2})$ time overall. The remaining lines of Algorithm \ref{alg:APSP-vcc} only perform broadcasts in Lines \ref{line:1}--\ref{line:2}, which takes $|R|=\tilde\BO(\sqrt{n})$ rounds. Therefore the total runtime is $\tilde\BO(\sqrt{n})$. 

The $(2 + o(1))$-approximation ratio for Algorithm \ref{alg:APSP-vcc} is immediately derived from \cite{nanongkai2014distributed}, as we do not change Nanongkai's algorithm besides executing it deterministically. 
\end{proof}

\section{Open Problems} \label{sec:open}

It is natural to ask whether our method of proving lower bounds for the diameter in the \vcc model can be extended to other problems. Of particular interest are those discussed in \cite{frischknecht2012networks}, since these problems use similar graph constructions for proving lower bounds. It would also be of interest to further reduce the runtime of approximating APSP in the \vcc and \ecc model, maybe also at the cost of larger approximation factors.

\addcontentsline{toc}{section}{References} 
\bibliographystyle{plain}
\bibliography{references}

\appendix
\section{Appendix}

\subsection{Proof of Lemma~\ref{diam2}}\label{app:diam2}

\begin{lem:d2}
The diameter of $G_{a, b}$ is at most $2+1/p$.
\end{lem:d2}
\begin{proof}
We show case by case that for any nodes $u$ and $v$ in $G_{a, b}$ the distance $d_w(u, v)$ is at most $2+1/p$. The cases are as follows:
\begin{enumerate}
\item \textbf{Nodes $u$ and $v$ are both in $G_{a}$:} Every node in $G_a$ other than $C_L$ is connected to $C_L$ by an edge of length $1$, and thus each node in $G_a$ can reach any other node in $G_a$ using at most two edges of length $1$. Thus, $d_w(u, v) \leq d_w(u, c_L) + d_w(c_L, v) \leq 2$.

\item \textbf{Nodes $u$ and $v$ are both in $G_{b}$:} This case is identical to the previous case, so $d_w(u, v) \leq 2$.

\item \textbf{Node $u$ is in $G_a$ and node $v$ is in $G_b$ (or vice verse):} From $u$ it is at most one hop to $C_L$ of length $1$, and from $v$ it is at most one hop to $C_R$ of length 1. Since the edge between $c_L$ and $c_R$ has weight $1/p$, we conclude that $d_w(u, v) \leq d_w(u, c_L) + d_w(c_L, c_R) + d_w(c_R, v) = 2+1/p$.
\end{enumerate}
\end{proof}

\subsection{Proof of Lemma~\ref{2approx2}}\label{app:2approx2}
\begin{lem:2a}
The diameter of $G_{a, b}$ is $1+1/p$ if the sets $a$ and $b$ are disjoint, else it is $2+1/p$.
\end{lem:2a}

\begin{proof}
\textbf{If inputs $a$ and $b$ are not disjoint,} then there exists an $i \in \{1, \dots, k(n)^2\}$ such that $a(i)=b(i)=1$. Let us fix such an $i$ for now and let $\nu:=i\mod k(n)$ and $\mu:=k(n)+\left\lfloor \frac{i}{k(n)}\right\rfloor$. We show that the two nodes $l_\nu$ and $r_\mu$ have distance of at least $2+1/p$. The path must contain an edge of length $1/p$ from the cut-set $C_{k(n)^2}$, since these are the only edges that connect $G_a$ to $G_b$. To obtain a path of length $1+1/p$ we are only allowed to add one more edge from either $G_a$ or $G_b$. When looking at the construction, the only two paths of length $1+1/p$ that we could hope for are $(l_\nu,l_\mu,r_\mu)$ and $(l_\nu,r_\nu,r_\mu)$. However, due to $a(i)=b(i)=1$ and the implied choice of $\nu$ and $\mu$, we know that the construction of $G_{a,b}$ does not include edge $(l_\nu,l_\mu)$ nor edge $(r_\nu,r_\mu)$. Thus none of these paths exists and we conclude that $d_w(l_\nu,r_\mu) \geq 2+1/p$.

\textbf{Conversely if $a$ and $b$ are disjoint,} the diameter of $G_{a,b}$ is at most $1+1/p$. We prove this by showing that for any nodes $u$ and $v$ in $G_{a,b}$ the distance $d_w(u,v)$ is at most $1+1/p$. To do this we distinguish three cases:
\begin{enumerate}
\item \textbf{Node $u$ is in $G_a$ and node $v$ is in $G_b$ (or vice verse):} When considering the nodes $\{c_L,c_R,w_0,w_1,w_2,\dots\}$, we notice that from each of these nodes each other node in the graph can be reached within $2$ hops. Now we can assume without loss of generality that $u=l_\nu\in L$ and $v=r_\mu\in R$ for some $\mu,\nu\in\{1, \dots, 2 k(n)-1\}$. Since we assumed that $a$ and $b$ are disjoint there must be either at least one of the edges $(l_\nu,l_\mu)$ or $(r_\nu,r_\mu)$ in case that one of the nodes is in \textbf{UP} and the other node is in \textbf{LP}.  Thus there is at least one of the paths $(l_\nu,l_\mu,r_\mu)$ or $(l_\nu,r_\nu,r_\mu)$ with $d_w(l_\nu,r_\mu)\leq 1+1/p$. In the remaining case $u,v$ are both in \textbf{UP} or both in \textbf{LP}, and we make use of the clique-edges and conclude that $u$ and $v$ are connected by path $(l_\nu,r_\nu,r_\mu)$ of length $d_w(l_\nu,r_\nu)+d_w(r_\nu,r_\mu)=1+1/p$.

\item \textbf{Nodes $u$ and $v$ are both in $G_a$:} Let $v = a_i$. In the above case, we can get from node $u$ to $b_i$ using a path of length $1+1/p$. Since the edge $(a_i, b_i)$ exists and has weight $1/p$, we can get from $u$ to $v$ using a path of length $1+1/p$.

\item \textbf{Nodes $u$ and $v$ are both in $G_b$:} same as the above case where both $u$ and $v$ are in $G_a$.
\end{enumerate}
Finally note, that these two cases combined with the upper bound from Lemma~\ref{diam2} imply that $d_w(l_\nu,r_\mu)=2+1/p$ if and only if $a$ and $b$ are not disjoint.
\end{proof}

\subsection{Proof of Lemma~\ref{greedyvc}}\label{app:greedyvc}

\begin{lem:gvc}
Given a graph $G$, if the smallest possible $k$-hitting set uses $N$ vertices, then $S$ contains at most $\BO(N \log n)$ vertices.
\end{lem:gvc}
\begin{proof}
We follow the proof of \cite{uiucnotes} that is originally stated for vertex covers and adapt it to $k$-hitting sets. Since the optimal solution $OPT$ uses $N$ nodes, there must exist some vertex that is contained in at least $\lceil n/N\rceil$ sets $S^k(u)$. Our greedy algorithm chooses the vertex $v_{\max}$ contained in as many sets $S^k(u)$ as possible. Thus $v_{\max}$ is contained in at least $\lceil n/N\rceil$ sets. After the first iteration of the algorithm, there are at most $\lfloor n(1 - 1/N) \rfloor$ sets $S^k(u)$ such that $S^k(u) \cap S = \emptyset$. Now observe, that $OPT$ is still a $k$-hitting set for the nodes in the remaining sets $S^k(u)$. By the same argument as above, since there is a $k$-hitting set that uses $N$ sets, there exists a vertex contained in at least $\lfloor n(1 - 1/N)\rfloor/N$ remaining sets $S^k(u)$. By induction, we can show that after $r$ rounds, there are at most $n(1 - 1/N)^r$ sets $S^k(u)$ that are disjoint from $S$. Choosing $r = \lceil N \ln n\rceil$ shows that $S$ is guaranteed to be a valid hitting set after $\BO(N \log n)$ rounds.
\end{proof}

\end{document}